\documentclass[11pt]{scrartcl}

\usepackage[a4paper, left=2.54cm, right=2.54cm, top=2.54cm, bottom=3.54cm]{geometry}

\usepackage{authblk} 

\usepackage[ngerman, UKenglish]{babel} 
\usepackage[utf8]{inputenc} 
\usepackage[T1]{fontenc}
\usepackage{comment}
\usepackage{verbatim} 

\usepackage[UKenglish]{isodate}

\usepackage{lmodern} 

\usepackage[pdfpagelabels, hidelinks]{hyperref} 
\usepackage{enumerate} 
\usepackage[onehalfspacing]{setspace}

\usepackage{natbib}

\usepackage{amsmath, amsthm, amssymb,amsfonts} 
\usepackage{dsfont} 
\usepackage{mathtools,thmtools}

\usepackage[table]{xcolor}
\usepackage{diagbox}

\allowdisplaybreaks

\newcommand\blfootnote[1]{
	\begingroup
	\renewcommand\thefootnote{}\footnote{#1}
	\addtocounter{footnote}{-1}
	\endgroup
}

\usepackage{xcolor} 

\usepackage{orcidlink}

\theoremstyle{plain} 
\newtheorem{theorem}{Theorem}[section]
\newtheorem{lemma}[theorem]{Lemma}

\newtheorem*{lemma*}{Lemma}
\newtheorem*{theorem*}{Theorem}
\newtheorem*{corollary*}{Corollary}
\newtheorem*{claim*}{Claim}

\theoremstyle{definition}
\newtheorem{definition}[theorem]{Definition}
\newtheorem{notation}[theorem]{Notation}
\newtheorem{remark-notation}[theorem]{Remark and Notation}
\newtheorem{example}[theorem]{Example}
\newtheorem{construction}[theorem]{Construction}
\newtheorem{remark}[theorem]{Remark}

\AtEndEnvironment{rem}{\hfill\ensuremath{\diamond}}

\newtheorem*{definition*}{Definition}
\newtheorem*{example*}{Example}
\newtheorem*{remark*}{Remark}

\theoremstyle{definition} 
\newtheorem{open-question}{Question}[section]

\numberwithin{equation}{section}

\title{Contributions to the hierarchy of\\probabilistic languages}

\author[i,ii]{Lothar Sebastian Krapp \orcidlink{0000-0003-3102-1923}\,}
\author[i] {Remo Nitschke \orcidlink{0000-0002-3660-1006}\,}
\affil[i]{\,Institut für Interdisziplinäre Sprachevolutionswissenschaft, Universität Zürich, Switzerland}
\affil[ii]{\,Fachbereich Mathematik und Statistik, Universität Konstanz, Germany}

\date{}

\begin{document}
	
	\pagenumbering{arabic}
	
	
	\maketitle
	\blfootnote{Math Subject Classification (2020): Primary 68Q45
    , 68Q42
    ; Secondary 68Q87
    , 68T50
    , 91F20
    , 03D05
    , 68Q70
    .}
    \blfootnote{Keystrings: probabilistic formal grammar, Chomsky hierarchy, complexity of languages.
	}
	\blfootnote{Corresponding Author: Lothar Sebastian Krapp, sebastian.krapp2@uzh.ch.
	}
	\vspace{-1.3cm}
	
	\begin{abstract}\small
		\noindent\textbf{Abstract.} We reconsider the theory of probabilistic formal languages generated by $n$-gram models and by probabilistic context-free grammars (PCFGs). The expected hierarchy of probabilistic grammars is established by proving that every probabilistic language generated by an $n$-gram model is also generated by some PCFG, while some probabilistic languages generated by PCFGs cannot be generated by any $n$-gram model.
        We introduce the notion of fully connected PCFGs, namely PCFGs in Chomsky normal form where every production rule only involving non-terminals has non-zero probability. Our main result shows that any probabilistic language generated by an $n$-gram model differs from any probabilistic language generated by a fully connected PCFG. Therefore, the class of probabilistic languages generated by $n$-gram models is not a subset of the class generated by \emph{fully connected} PCFGs.
	\end{abstract}
	

	
\section{Introduction}\label{sec:intro}
		
	Classically, Formal Language Theory considers formal grammars and the formal languages they generate. Roughly, a formal grammar consists of a vocabulary of symbols, a set of non-terminal symbols, and a set of production rules over the prior. A formal grammar generates a formal language, which is a set of possible sequences. Most prominently, in the hierarchy of formal grammars (also known as Chomsky hierarchy) the two lowest levels are regular grammars and context-free grammars (see \cite{Chomsky1956, CHOMSKY1959, chomskyMiller, Chomsky1963}), which are the mechanisms under discussion in this work. A language $L$ generated by a formal grammar $G$ is the set of all sequences which can be obtained by the production rules specified in $G$. Context-free grammars are more expressive in the sense that there are languages that are generated by context-free grammars but cannot be generated by regular grammars. Since the hierarchy of grammar is a hierarchy of set-inclusions, any regular grammar is also context-free by definition, so any regular language can be generated by a context-free grammar. However, for any regular language $L$ there is also a context-free grammar $G$ that is not regular such that $G$ generates $L$.
    Two distinct grammars that generate the same language $L$ but with potentially different parses are considered \textit{weakly equivalent} (see \cite{chomskyMiller}).
	
	In the classical formal language theoretic framework described above, the symbols and sequences that constitute a language are considered qualitatively, not quantitatively. More precisely, as soon as there is one possible derivation for a given sequence from the rules of the grammar, the sequence is part of the language. Thus, whether a sequence has exactly one derivation or several ones does not affect that it becomes a part of the generated language, neither does it matter how probable the sequence may be.
	
	In this work, we move away from this qualitative framework by working with \emph{probabilistic} (formal) languages. Rather than treating all production rules as ``equally likely'' for the production of a sequence, we assign probabilities to each such rule. Thus, each sequence that is generated by a probabilistic (formal) grammar is assigned a probability, where the sum of all probabilities for the sequences of symbols in the language equals $1$.
	
	Probabilistic grammars have wide applications from generating text (see \cite{jurafsky}) to deriving the probability of sequences for information theoretic measures (see \cite{SMITHLevy}) or to infer probabilistic dependencies between symbols in a dataset in a variety of fields (see \cite{bosshard204}). Probabilistic context-free grammars specifically have played a role in parsing theory, for instance when the most likely parse-tree for a given sequence has to be calculated (see e.g.\ \cite[Chapter 14]{jurafsky}). Finally, these grammars have been used as proxies for human syntax (see \cite{Jin2021}).
	
	In the context of classical formal grammars, one would usually ask a question like: ``Given this language $L$, is there a grammar $G$ with certain properties (like `being regular') that generates the language $L$?'' We ask this question for probabilistic grammars to find dividing lines between their ``probabilistic expressive power''. We focus on $n$-gram models as probabilistic generalisations of (right-)regular grammars and on probabilistic context-free grammars (PCFGs), more specifically focusing on such that are fully connected and in Chomsky normal form (PCFG$_{\text{FC}}$). We show that the following inclusions of classes of probabilistic languages are strict:
    \begin{center}
    \includegraphics[width=15cm]{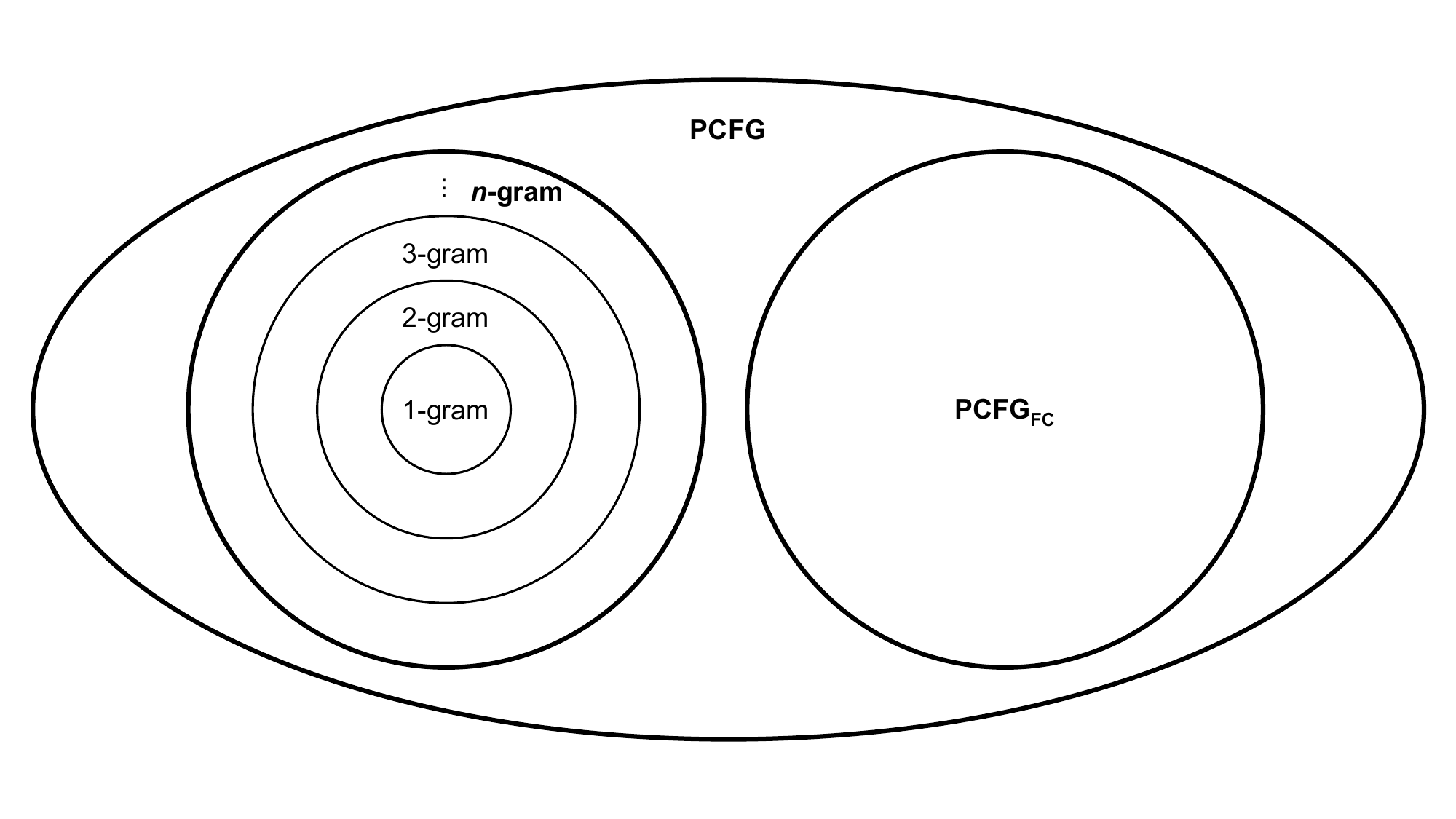}
    \end{center}
	More precisely, we show that any probabilistic language of a $1$-gram model can be generated by some $2$-gram model etc.\ and every probabilistic language of some $n$-gram model can also be generated by some PCFG. Strictness of these inclusions is verified by giving examples of certain probabilistic languages, e.g.\ of a PCFG generating a probabilistic language that is not generated by any $n$-gram model. Finally, our main result shows that the class of probabilistic languages generated by $n$-gram models is disjoint from the class generated by fully connected PCFGs.
   
    Just as there can be weakly equivalent discrete grammars, 
    there can be weakly equivalent probabilistic grammars. These are probabilistic grammars that generate the same probabilistic language (a probability distribution over sequences, see Section \ref{sec:probL}) with different parses. Our main result shows that fully connected PCFGs and $n$-gram models can never be weakly equivalent. They may allow the same sequences, but they will represent differing probability distributions over the possible sequences. This opens up new avenues of probing the likely underlying class of grammar of a corpus of sequences.

    This result is relevant as it allows considerations of \textit{how likely} it is that a set of data was generated by probabilistic context-free mechanism or a probabilistic regular mechanism. Under a discrete approach, this is not an answerable question (see \cite{Greibach1968}).

\section{Preliminaries on probabilistic languages}
\label{sec:probL}
	
	We denote by $\mathbb{N}$ the set of natural number without $0$ and for any $\ell \in \mathbb{Z}$ we let $[\ell]=\{1,\ldots,\ell\}$ if $\ell\geq 1$ and $[\ell]=\emptyset$ if $\ell\leq 0$.
	Throughout this work, we fix an alphabet $\Sigma=\{a_1,a_2,\ldots,a_k\}$ for some $k\in \mathbb{N}$. For notational convenience, we set $a=a_1$ and $b=a_2$. The alphabet $\Sigma$ also serves as the set of terminal symbols for the PCFGs we consider. We fix a symbol $\#$ that is not contained in $\Sigma$; it is later used as the stopping symbol for $n$-gram models. 
	
	We generally follow the basic notations and terminology on formal languages as introduced in \cite[Section~2.1]{wirth}.
	The set $\Sigma^+$ consists of all non-empty strings over the alphabet $\Sigma$. For a string $w\in \Sigma^+$, its length is denoted by $|w|$.  
	We use standard notation for powers and concatenations, e.g.\ $\Sigma^\ell$ consists of all strings of length $\ell$, the set $b\Sigma^\ell$ consists of all strings of  length $\ell+1$ starting with $b$, and $ba^\ell\in b\Sigma^\ell$ is the string $b\underbrace{a\ldots a}_{\ell \text{ times}}$ of length $\ell+1$. 
	We denote the empty string of length $0$ by $\varepsilon$. The result of any empty concatenation results in $\varepsilon$, and for any symbol $x\in \Sigma$, we have $x^0=\varepsilon$. We use the convention $\#^0 = \#$ for later notational convenience when introducing unigram models.
	We will not consider grammars that may generate empty strings.\footnote{Our general framework can be adjusted to grammars generating the empty string. This would, however, also result on some further changes to the technicalities of $n$-gram models and PCFGs that we introduce.}
	
	Throughout, we use a set of non-terminal symbols $\mathcal{V}=\{X_1,\ldots,X_m\}$ for some $m\in \mathbb{N}$. Again, we also set $A=X_1$ and $B=X_2$. While most probability measures are simply denoted by $P$, we add an index if the set on which this measure is defined may be ambiguous.
	
	As explained in the \nameref{sec:intro}, we focus on probabilistic languages that allow for any possible non-empty finite strings over $\Sigma$. The languages that we consider are therefore purely distinguished by the probability distribution on $\Sigma^+$. 
	
	\begin{definition}\label{def:problang}
		A \textbf{probabilistic language} (over the alphabet $\Sigma$) is a probability distribution $P$ on $\Sigma^+$ that assigns to every string a positive probability. In other words, a probabilistic language is a map
		$$P\colon \Sigma^+ \to (0,1)$$
		with $\sum_{w\in\Sigma^+} P(w) =1$.
	\end{definition}
	
	The probabilistic languages we consider in this work all arise from production systems, which we call \textbf{probabilistic grammars}. We introduce the relevant probabilistic grammars in the following.

	\begin{definition}\label{def:ngram}
		Let $n\in \mathbb{N}$. We set\footnote{The set $\Sigma_n^+ = \{\#^{n-1}\}\cup \#^{n-2}\Sigma \cup \ldots \cup \#\Sigma^{n-2}\cup \Sigma^{n-1}$ is the state space of the $n$-th order Markov chain underlying the $n$-gram model.} $$\Sigma_n^+ = \{\#^{n-1}\}\cup \bigcup_{\ell=2}^{n-1} \#^{n-\ell} \Sigma^{\ell-1} \cup \Sigma^{n-1}.$$ An \textbf{$n$-gram model} 
		is specified by giving,	for every $s\in \Sigma_n^+$, a probability distribution $P$ on the set of transitions $\{s\to x\mid x\in \Sigma\cup\{\#\}\}$, satisfying
		\begin{itemize}
			\item $P(\#^{n-1}\to \#) = 0$ (the empty string receives probability $0$), and
			
			\item $P(s\to x)>0$ for any $(s,x)\in \Sigma_n^+ \times (\Sigma\cup\{\#\})$ with $(s,x)\neq(\#^{n-1},\#)$.
		\end{itemize} The probabilistic language generated by this $n$-gram model is given by the following probability distribution on $\Sigma^+$:
		 If $w=x_1\ldots x_\ell \in \Sigma^\ell$ for $\ell\in [n-2]$, then
			\begin{align*}P(w)&=P(\#^{n-1}\to x_1)\cdot P(\#^{n-2}x_1\to x_{2})\cdot \ldots \\&\phantom{=\ }\cdot P(\#^{n-\ell}x_1\ldots x_{\ell-1}\to x_\ell)\cdot P(\#^{n-\ell-1}x_1\ldots x_{\ell}\to \#).\end{align*}
			If $w=x_1\ldots x_\ell\in \Sigma^\ell$ for $\ell\geq n-1$, then
			\begin{align*}P(w)&=P(\#^{n-1}\to x_1)\cdot P(\#^{n-2}x_1\to x_{2})\cdot \ldots \\&\phantom{=\ }\cdot P(x_1\ldots x_{n-1}\to x_n)\cdot P(x_2\ldots x_{n}\to x_{n+1}) \cdot \ldots
			 \\&\phantom{=\ }\cdot P(x_{\ell+1-n}\ldots x_{\ell-1}\to x_\ell)\cdot P(x_{\ell+2-n}\ldots x_{\ell}\to \#). \end{align*}
	\end{definition}
	
	\begin{remark}
			 In \Autoref{def:ngram} we stipulate ``$P(s\to x)>0$ for any $(s,x)\in \Sigma_n^+ \times (\Sigma\cup\{\#\})$ with $(s,x)\neq(\#^{n-1},\#)$''. 
			 This ensures that the Markov chain is as ``fully connected'' as reasonably possible, i.e.\ all transitions that are theoretically possible are required to have a non-zero probability. 
			 Indeed, if there were some $s\in \Sigma_n^+$ and $x\in \Sigma\cup\{\#\}$ with $(s,x)\neq (\#^{n-1},\#)$ such that $P(s\to x)=0$, this would give rise to a string $w\in \Sigma^+$ with $P(x)=0$. This string $w$ is obtained from the string $sx$ by deleting all occurrences of the symbol $\#$.
	\end{remark}

	As unigram models do not rely on transitional probabilities, we point out in the following some notational conventions.
	
	\begin{notation}
		Note that $\Sigma_1^+ = \{\#^0\}\cup \Sigma^{0}=\{\#,\varepsilon\}$. A unigram model is therefore specified
		by a probability distribution on $\{\#\to x\mid x\in \Sigma\}$ and a probability distribution on $\{\varepsilon\to x\mid x\in \Sigma\cup\{\#\}\}$, satisfying
		$P(s\to x)>0$ for any $(s,x)\in \{\#,\varepsilon\} \times (\Sigma\cup\{\#\})$ with $(s,x)\neq(\#,\#)$. The probabilistic language generated by this unigram model is given by the following probability distribution on $\Sigma^+$:
		If $w=x_1\ldots x_\ell\in \Sigma^\ell$ for $\ell\in \mathbb{N}$, then
		\begin{align*}P(w)&=P(\#\to x_1)\cdot P(\varepsilon \to x_{2})\cdot \ldots \cdot  P(\varepsilon \to x_{\ell})\cdot  P(\varepsilon \to \#). \end{align*}
		We will later sometimes omit the empty string $\varepsilon$ and simply write $P(\to x)$ instead of $P(\varepsilon\to x)$ for any $x\in \Sigma$.
	\end{notation}
	
	Recall that a \textbf{probabilistic context-free grammar (PCFG)} consists of a set of non-terminal symbols, a set of terminal symbols, a set of (production) rules, each of which is endowed with a probability, and a start symbol (see \cite[Section~14.1]{jurafsky}). The probability of a given parse tree $T$ is simply the product of all the product rules that are applied. Moreover, the probability for a certain sequence is the sum over all probabilities of parse trees yielding this sequence. We follow this framework with the exception that we do not fix a start symbol. We make this precise in the following.
	
	\begin{definition}\label{def:pcfg}
		A PCFG (without fixed starting symbol) consists of a probability distribution on the set of possible starting symbols $\{\to X\mid X\in \mathcal{V}\}$ and for each $X\in\mathcal{V}$ a probability distribution on all production rules\footnote{We use the convention that no empty string can be generated by applying a rule to a non-terminal symbol.} $\{X\to \beta \mid \beta \in (\Sigma\cup \mathcal{V})^+\}$. The probability of a sequence $w\in \Sigma^+$ in the probabilistic language generated by this PCFG is given by
		$$P(w)=\sum_{X\in \mathcal{V}}P(\to X)\sum_{T}P(T),$$
		where $T$ in the sum above ranges over all parse trees which have $X$ as starting symbol (root node) and generate the sequence $w$.
	\end{definition}
	
	The PCFGs which we focus on in this work are in Chomsky normal form (i.e.\ only production rules of the form $X\to YZ$ and $X\to x$ for $X,Y,Z\in \mathcal{V}$ and $x\in \Sigma$ are allowed) and fully connected (i.e.\ each production rule  $X\to YZ$ for some $X,Y,Z\in \mathcal{V}$ has non-zero probability). We make this precise in the following.
	
	\begin{definition}
		A fully connected probabilistic context-free grammar in Chomsky normal form, which we simply call a \textbf{fully connected PCFG}, consists of a probability distribution on the set of starting symbols $\{\to X\mid X\in \mathcal{V}\}$, and for each $X\in\mathcal{V}$ a probability distribution on all production rules $\{X\to YZ\mid Y,Z\in \mathcal{V}\}\cup\{X\to x\mid x\in \Sigma\}$ with $P(X\to YZ)>0$ for any $Y,Z\in \mathcal{V}$. Moreover, for any terminal $x\in \Sigma$ there must be at least one non-terminal $X\in\mathcal{V}$ such that $P(X\to x)>0$.
		The probabilistic language generated by this fully connected PCFG is given by the probability distribution on $\Sigma^+$ defined via parse trees as usual (see \Autoref{def:pcfg}). 
	\end{definition}
	
	\begin{remark}\label{rmk:consistentgrammars}
        \begin{enumerate}
            \item The discrete (i.e.\ non-probabilistic) counterparts of the grammars we consider are regular grammars and context-free grammars in Chomsky normal form, where each production rule of the form $X\to xY$ (for regular grammars) and $X\to YZ$, $X\to x$ (for context-free grammars) are included. 
            Here, $X,Y,Z$ are any non-terminal symbols and $x$ is a terminal symbol. These formal grammars will always produce the whole possible language of all non-empty strings $\Sigma^+$. Thus, they are always trivially weakly equivalent.

            \item Fully connected $n$-gram models and fully connected PCFGs ensure (in practice) that any possible non-empty string gets evaluated with a positive probability. This is crucial for likelihood calculations, which we will elaborate in \Autoref{rmk:likelihoods}.
            
            \item One crucial assumption on all probabilistic grammars we consider is that the equations for the sequence probabilities $P(w)$ given in the definitions of their probabilistic languages actually define probability measures. In other words, all probabilistic grammars have to be consistent (cf.\ \cite[page 495]{jurafsky}; see \cite[page 388] {manning} 
		      for an example of a fully connected PCFG that is inconsistent).
		      Moreover, in our specific context we also need all probabilities to be non-zero on $\Sigma^+$. For fully connected PCFGs this is ensured by the condition that for any $x\in \Sigma$ there is some $X\in\mathcal{V}$ with $P(X\to x)>0$. This condition alone does not ensure that a consistent PCFG in Chomsky normal form generates a probabilistic language in our sense, as \Autoref{ex:noproblanguage} below illustrates.
        \end{enumerate}
		
	\end{remark}

	\begin{example}\label{ex:noproblanguage}
		We present a consistent PCFG in Chomsky normal form that does not generate a probabilistic language. Let $k=m=2$, i.e.\ $\Sigma =\{a,b\}$ and $\mathcal{V}=\{A,B\}$.
		Consider the PCFG given by $P(\to A)=1$, $P(A\to AB)=P(A\to a)=\frac 12$ and $P(B\to b)=1$. While this is a consistent PCFG with $P(ab^\ell)=\frac{1}{2^\ell}$ for any $\ell\in \mathbb{N}$, it cannot generate the sequence $b$. Thus, $P(b)=0$ showing that this PCFG does not adhere to \Autoref{def:problang}. Note that this PCFG is not fully connected.
	\end{example}
	
	Throughout the rest of this work, all probabilistic grammars we consider are assumed to define a probabilistic language as introduced in \Autoref{def:problang}.
	
\section{Hierarchy of probabilistic grammars}
		
	In this section, we show the inclusion 
	\begin{align}\text{unigram}\subsetneq \text{bigram}\subsetneq \text{trigram}\subsetneq\ldots \subsetneq n\text{-gram}\subsetneq \ldots\subsetneq\text{PCFG}. \label{eq:inclusions}\end{align}	
	We will proceed as follows:
	\begin{itemize}
		\item \Autoref{ex:ngraminclusion} establishes inclusions for all $n$-gram models. More precisely, for a given $n$-gram model we construct an $(n+1)$-gram model that generates the same probabilistic language as the $n$-gram model.
		
		\item \Autoref{ex:exclusionngram} shows that the inclusions for $n$-gram models are strict (for $n\geq 2$). It does so by providing, for each $n\in \mathbb{N}$, an $(n+1)$-gram model whose probabilistic language cannot be generated by any $n$-gram model. Special care has to be taken for unigram models. For the strict inclusion, the alphabet $\Sigma$ needs to have at least two distinct symbols (see \Autoref{ex:exclusionngramsimple}).
		
		\item \Autoref{ex:unigrampcfg} (for $n=1$) and \Autoref{ex:ngrampcfg} (for $n\geq 2$) establish that the probabilistic language generated by any $n$-gram model can already be generated by some PCFG, as long as we have control over the number of non-terminal symbols. 
		
		\item Finally in \Autoref{ex:fullyconnectedsimple}, we give a simple example of fully connected PCFG only using one non-terminal symbol. 
        Due to our main result \Autoref{thm:main}, this PCFG generates a probabilistic language that cannot be generated by any $n$-gram model. 
	\end{itemize}

	\begin{construction}\label{ex:ngraminclusion}
		Let $n\in \mathbb{N}$ and let $G$ be an $n$-gram model generating the probabilistic language $P$. 
		We construct an $(n+1)$-gram model $H$  also generating the probabilistic language $P$. In order to distinguish between the probabilities given in $G$ and $H$, we denote them by $P_G$ and $P_H$, respectively. For any $x_1,\ldots,x_{n+1} \in \Sigma\cup\{\#\}$ with $x_1\ldots x_n \in \Sigma_{n+1}^+$, we set
		$$P_H(x_1\ldots x_n \to x_{n+1})=
		\begin{cases}
			P_G(x_2\ldots x_n \to x_{n+1}),& \text{ if }n\geq 2\\
			P_G(x_1^0 \to x_{2}),&\text{ if }n=1.
		\end{cases}
		$$ 
		We verify that the $(n+1)$-gram model $H$ generates the same probabilistic language as $G$.
		Let $w=x_1\ldots x_\ell\in \Sigma^+$. We first consider the case $n\geq 2$.
		If $\ell\in [n-2]$, then
		\begin{align*}
			P_H(w)&=P_H(\#^{n}\to x_1)\cdot P_H(\#^{n-1}x_1\to x_{2})\cdot \ldots \\&\phantom{=\ }\cdot P_H(\#^{n+1-\ell}x_1\ldots x_{\ell-1}\to x_\ell)\cdot P_H(\#^{n+1-\ell-1}x_1\ldots x_{\ell}\to \#)\\
			&=		P_G(\#^{n-1}\to x_1)\cdot P_G(\#^{n-2}x_1\to x_{2})\cdot \ldots \\&\phantom{=\ }\cdot P_G(\#^{n-\ell}x_1\ldots x_{\ell-1}\to x_\ell)\cdot P_G(\#^{n-\ell-1}x_1\ldots x_{\ell}\to \#)\\
			&=P_G(w).
			\end{align*}
		If $\ell = n-1$, then
		\begin{align*}
			P_H(w)&=P_H(\#^{n}\to x_1)\cdot P_H(\#^{n-1}x_1\to x_{2})\cdot \ldots \\&\phantom{=\ }\cdot P_H(\#^{n+1-\ell}x_1\ldots x_{\ell-1}\to x_\ell)\cdot P_H(\#^{n+1-\ell-1}x_1\ldots x_{\ell}\to \#)\\
			&=		P_G(\#^{n-1}\to x_1)\cdot P_G(\#^{n-2}x_1\to x_{2})\cdot \ldots \\&\phantom{=\ }\cdot P_G(\#^{n-\ell}x_1\ldots x_{\ell-1}\to x_\ell)\cdot P_G(x_1\ldots x_{\ell}\to \#)\\
			&=P_G(w).
		\end{align*}
		If $\ell \geq n$, then
		\begin{align*}
			P_H(w)&=P_H(\#^{n}\to x_1)\cdot P_H(\#^{n-1}x_1\to x_{2})\cdot \ldots \cdot P_H(\#x_1\ldots x_{n-1} \to x_{n})  \\&\phantom{=\ }\cdot P_H(x_1\ldots x_{n}\to x_{n+1})\cdot P_H(x_2\ldots x_{n+1}\to x_{n+2}) \cdot \ldots
			\\&\phantom{=\ }\cdot P_H(x_{\ell-n}\ldots x_{\ell-1}\to x_\ell)\cdot P_H(x_{\ell+1-n}\ldots x_{\ell}\to \#)\\
			&=P_G(\#^{n-1}\to x_1)\cdot P_G(\#^{n-2}x_1\to x_{2})\cdot \ldots \cdot P_G(x_1\ldots x_{n-1} \to x_{n})  \\&\phantom{=\ }\cdot P_G(x_2\ldots x_{n}\to x_{n+1})\cdot P_G(x_3\ldots x_{n+1}\to x_{n+2}) \cdot \ldots
			\\&\phantom{=\ }\cdot P_G(x_{\ell-n+1}\ldots x_{\ell-1}\to x_\ell)\cdot P_G(x_{\ell+2-n}\ldots x_{\ell}\to \#)\\
			&=P_G(w).
		\end{align*}
		Now we consider the case $n=1$ and obtain
		\begin{align*}
			P_H(w)&=P_H(\#\to x_1)\cdot P_H(x_1\to x_{2})\cdot \ldots \cdot P_H(x_{\ell-1}\to x_\ell)\cdot P_H(x_{\ell}\to \#)\\
			&=P_G(\#\to x_1)\cdot P_G(\to x_{2})\cdot \ldots \cdot P_G(\to x_\ell)\cdot P_G(\to \#)\\
			&=P_G(w),
		\end{align*}
		as required.
	\end{construction}
	
	We next present a minimalistic examples of an $(n+1)$-gram model whose probabilistic language cannot be generated by any $n$-gram model. We point out that due to \Autoref{ex:ngraminclusion}, the probabilistic language can also not be generated by any $r$-gram model with $r\leq n$.

	\begin{example}\label{ex:exclusionngram}
		Let $n\in \mathbb{N}$ with $n\geq 2$.
		Consider the $(n+1)$-gram model $G$ given by the following probability distributions:
		\begin{itemize}
			\item $P(\#^n\to x)= \frac 1k$ for any $x\in \Sigma$;
			
			\item $P(\#^\ell s\to \#)=\frac 34$, and $P(\#^\ell s\to x)=\frac 1{4k}$ for any $x\in \Sigma$,  $s\in \Sigma^{n-\ell}$ and $1\leq \ell\leq n-1$;
			
			\item $P(s\to x)=\frac 1{3(k+1)}$ for any $s\in \Sigma^n$ and $x\in \Sigma\cup\{\#\}$.
		\end{itemize}
		Assume, for a contradiction, that $H$ were an $n$-gram model generating the same probabilistic language as $G$.
		Set $C=P(a^{n-1})=P_H(a^{n-1})$. We have
		\begin{align*}
			CP_H(a^{n-1}\to a) &= \frac{P_H(a^{n-1})}{P_H(a^{n-1}\to \#)}\cdot P_H(a^{n-1}\to a)\cdot P_H(a^{n-1}\to \#)\\
			&=P_H(a^n)=P(a^n)\\
			&=\frac{P(a^{n-1})}{P(\#a^{n-1}\to \#)}\cdot P(\#a^{n-1}\to a)\cdot P(a^{n}\to \#)\\
			&=\frac{C}{\frac 34}\cdot \frac 1{4k}\cdot \frac 1{3(k+1)}\\
			&=\frac{C}{9k(k+1)}.
		\end{align*}
		Thus, $P_H(a^{n-1}\to a)=\frac 1{k(k+1)}$. Further,
		\begin{align*}
		\frac1{9k(k+1)} {P(a^n)} &= P_H(a^n)\cdot P_H(a^{n-1}\to a)\\
		&= \frac{P_H(a^{n})}{P_H(a^{n-1}\to \#)}\cdot P_H(a^{n-1}\to a)\cdot P_H(a^{n-1}\to \#)\\
		&= P_H(a^{n+1}) = P(a^{n+1})\\
		&= \frac{P(a^{n})}{P(a^{n}\to \#)}\cdot P(a^{n}\to a)\cdot P(a^{n}\to \#)\\
		&= P(a^{n})\cdot P(a^{n}\to a)\\
		&= \frac 1{3(k+1)} P(a^n),
		\end{align*}
		yielding the desired contradiction.
	\end{example}
	
	For unigram models we have to adapt \Autoref{ex:exclusionngram} and require an alphabet of size $k\geq 2$, as we observe in the following remark.
	
	\begin{remark}
		Suppose that $k=1$, i.e.\ $\Sigma = \{a\}$, and let $G$ be a bigram model. Then $G$ is specified by the two probabilities $P_G(a\to a)$ and $P_G(a\to \#)$, as $P(\#\to a)=1$. Let $H$ be the unigram model given by $P_H(\#\to a)=1$, $P_H(\to a)=P_G(a\to a)$ and $P_H(\to \#)=P_G(a\to \#)$. Then it is straight-forward to verify that $G$ and $H$ generate the same probabilistic language.
	\end{remark}
	
	\begin{example}\label{ex:exclusionngramsimple}
		Suppose that $k\geq 2$. Consider the bigram model $G$ given by probabilities:
			\begin{itemize}
				\item $P(\#\to x)= \frac 1k$ for any $x\in \Sigma$;
				
				\item $P(a\to \#)=\frac 23$ and $P(a\to x)=\frac1 {3k}$ for any $x\in \Sigma$;
				
				\item $P(x\to y)=\frac1{k+1}$ for any $x\in (\Sigma\setminus\{a\})$ and $y\in \Sigma\cup\{\#\}$.
			\end{itemize}
		Let $H$ be a unigram model and assume, for a contradiction, that $H$ generates the same probabilistic language as $G$. Then
			$$P_H(aab)=P_H(\#\to a)P_H(\to a)P_H(\to b)=P_H(\#\to a)P_H(\to b)P_H(\to a)=P_H(aba).$$
		Thus,
		\begin{align*}
			\frac1k \cdot \frac1{3k}\cdot \frac1{3k}\cdot \frac1{k+1}&=P(\#\to a)P(a\to a)P(a\to b)P(b\to \#)\\
			&=P(aab)=P(aba)\\
			&=P(\#\to a)P(a\to b)P(b\to a)P(a\to \#)\\
			&=\frac1k \cdot \frac1{3k}\cdot \frac1{k+1}\cdot \frac 23.
		\end{align*}
		yielding the contradiction $\frac 1{3k}=\frac 23$.
	\end{example}

	Next, we compare $n$-gram models with PCFGs. We first show that any language an $n$-gram model generates is already generated by a PCFG in Chomsky normal form.

	\begin{construction}\label{ex:unigrampcfg} 
		Let $G$ be an unigram model. We construct a PFCG $H$ over the set of non-terminals
		$$\mathcal{V}=\{Y_x\mid x\in \Sigma\}\cup\{Z_x\mid x\in \Sigma\}.$$
		The probabilities to specify $H$ are given as follows for any $x,y\in \Sigma$:
		\begin{itemize}			
			\item $P_H(\to Y_x)=P_G(\#\to x)$;
			
			\item $P_H(Y_x\to Z_x Y_y)=P_G(\to y)$;
			
			\item $P_H(Y_x\to x)=P_G(\to \#)$;
			
			\item $P_H(Z_x\to x) = 1$.
		\end{itemize}
		Let $w=x_1\ldots x_\ell \in \Sigma^+$. Then
		\begin{align*}
			P_H(w)&=P_H(\to Y_{x_1})\cdot P_H(Y_{x_1}\to Z_{x_1} Y_{x_2})P_H(Z_{x_1}\to x_1)\cdot \ldots \\
			&\phantom{=}\cdot P_H(Y_{x_{\ell-1}}\to Z_{x_{\ell-1}} Y_{x_\ell})P_H(Z_{x_{\ell-1}}\to x_{\ell-1})\cdot P_H(Y_{x_\ell}\to x_\ell)\\
			&=P_G(\#\to x_1)\cdot P_G(\to x_2)\cdot 1\cdot \ldots \cdot P_G(\to x_\ell)\cdot 1\cdot P_G(\to \#)\\
			&=P_G(w).
		\end{align*}
		Thus, $G$ and $H$ generate the same probabilistic language.
	\end{construction}
	
	\begin{construction}\label{ex:ngrampcfg} 
		Let $n\in \mathbb{N}$ with $n\geq 2$ and let $G$ be an $n$-gram model. We construct a PFCG $H$ over the set of non-terminals
		$$\mathcal{V}=\{Q_{w}\mid w\in \Sigma^{n-1}\}\cup\{Y_w\mid w\in \Sigma^+, |w|\leq n-1\}\cup\{Z_x\mid x\in \Sigma\}.$$
		We set
		\begin{itemize}
			\item $P(\to Y_w)=P_G(w)$ for any $w\in \Sigma^+$ with $|w|\leq n-1$,
			
			\item $P(\to Q_{x_1\ldots x_{n-1}})=P_G(\#^{n-1}\to x_1)\cdot \ldots\cdot P_G(\#x_1\ldots x_{n-2}\to x_{n-1})$ for any $x_1,\ldots,x_{n-1}\in \Sigma$,
			
			\item $P(Y_{x_1\ldots x_\ell} \to Z_{x_1}Y_{x_2\ldots x_\ell}) = 1$, for any $\ell\geq 2$ and $x_1,\ldots,x_{\ell}\in \Sigma$,
			
			\item $P(Y_{x}\to x) = P(Z_x\to x) = 1$ for any $x\in \Sigma$,

			\item $P(Q_{x_1\ldots x_{n-1}}\to Z_{x_1}Q_{x_2\ldots x_{n-1}y})=P_G(x_1\ldots x_{n-1}\to y)$ for any $x_1,\ldots,x_{n-1},y\in \Sigma$,
			
			\item $P(Q_{x_1\ldots x_{n-1}}\to Z_{x_1}Y_{x_2\ldots x_{n-1}})=P_G(x_1\ldots x_{n-1}\to \#)$ for any $x_1,\ldots,x_{n-1}\in \Sigma$.
		\end{itemize}
		It is a technical but straight-forward procedure to verify that $G$ and $H$ generate the same probabilistic language.
	\end{construction}

	\begin{example}\label{ex:fullyconnectedsimple}
		Let $k=m=1$ and let $G$ be the PCFG given by $P(\to A)=1$, $P(A\to AA)=P(A\to a)=\frac 12$. Then clearly $L=\{a\}^+=\{a,aa,aaa,\ldots\}$. To compute $P(a^r)$ one has to count the binary parse trees that yield $a^r$, which is exactly the $(r-1)$-th Catalan number
        $$C_r = \frac{(2(r-1))!}{r!(r-1)!}.$$
        Each such parse tree $T$ has probability $P(T)=\frac{1}{2^{r-1}}\cdot \frac{1}{2^r} = \frac{1}{2^{2r-1}}$, as the production rule $A\to AA$ has to be applied $r-1$ times, and the production rule $A\to a$ has to be applied $r$ times. 
        See  Figure~\ref{eq:prob_func} for further details.	
	\end{example}

\section{Disjoint probabilistic grammars}

	In this section, we prove our main result \Autoref{thm:main} that the probabilistic language of any $n$-gram model is distinct from the probabilistic language of any fully connected PCFG.

	\begin{lemma}\label{lem:powers}\label{lem:powers2}
		Let $n\in \mathbb{N}$. Then for any $n$-gram model and any $x\in \Sigma$,
		there are some $c,d\in (0,\infty)$ and $r\in (0,1)$ such that for any $\ell\geq n$, 
		$$P(x^\ell)=cr^{\ell},$$
		and for any $\ell\in \mathbb{N}$, 
		$$P(x^\ell)\geq dr^{\ell}.$$
	\end{lemma}
	
	\begin{proof}
		Set $r=P(x^{n-1}\to x)$ and $c= {P(x^{n})}P(x^{n-1}\to x)^{-n}$.
		Then
		\begin{align*}
			P(x^\ell)= \frac{P(x^{n})}{P(x^{n-1}\to \#)}P(x^{n-1}\to x)^{\ell-n}P(x^{n-1}\to \#) = {P(x^{n})}P(x^{n-1}\to x)^{\ell-n}=cr^\ell.
		\end{align*}
		It now suffices to set
		$$d=\min\left(\{c\}\cup\left\{\left.\tfrac{P(x^i)}{r^i}\ \right|\ i\in[n-1]\right\}\right).$$
	\end{proof}
	
	\begin{notation}
		For a given PCFG, we use the following notation. For $\alpha,\beta,\gamma\in [m]$ set $s_\alpha=P(\to X_\alpha)$, $t_{\alpha\beta\gamma}=P(X_\alpha \to X_\beta X_\gamma)$ and $u_{\alpha}=P(X_\alpha\to a)$. For any $\ell\in \mathbb{N}$, denote by $p_{\alpha,\ell}$ to be the conditional probability of obtaining the sequence $a^\ell$ from a parse tree whose root node is $X_\alpha$.
	\end{notation}
	
	The proof of the following lemma is a straight-forward calculation on parse tree probabilities.
	
	\begin{lemma}\label{lem:iterprob}
		For any fully connected PCFG, any $\alpha\in [m]$ and any $\ell\geq2$, we have the following:
		\begin{align}
			\notag p_{\alpha,1}&= u_\alpha,\\
			p_{\alpha,\ell} &= \sum_{\beta,\gamma\in [m]} t_{\alpha\beta\gamma} \sum_{i=1}^{\ell-1}p_{\beta,i}p_{\gamma,\ell-i},\label{lem:iterprob:3}\\
			P(a^\ell)&=\sum_{\alpha\in [m]}s_\alpha p_{\alpha,\ell}.\label{lem:iterprob:4}
		\end{align}
	\end{lemma}
	
	\begin{theorem}\label{thm:main}
		Let $n\in \mathbb{N}$. Then for any $n$-gram model $G$ and any fully connected PCFG $H$ their generated probabilistic languages $P_G$ and $P_H$ do not coincide, i.e.\ $P_G\neq P_H$.
	\end{theorem}
	
	\begin{proof}
		Assume, for a contradiction, that both models $G$ and $H$ generate the same probabilistic language. We therefore do not have to distinguish between the probability distributions $P_G$ and $P_H$ on $\Sigma^+$.
		By \autoref{lem:powers2}, there are $c,d\in (0,\infty)$ and $r\in (0,1)$ such that
		$P(a^\ell)=cr^{\ell}$ for any $\ell\geq n$, and
		$P(a^\ell)\geq dr^{\ell}$ for any $\ell\in \mathbb{N}$.
		Turning to the PCFG $H$, we show that there is some $D\in (0,\infty)$ such that for any $\ell\in \mathbb{N}$ we have
		$$P(a^\ell)\geq D(\ell-1)r^\ell.$$
		This leads to the required contradiction, as we then obtain
		$$cr^{\ell} \geq D(\ell-1)r^\ell$$
		for any $\ell\geq n$, which is not possible, as $c$ and $D$ are positive constants.
		
		Set 
		$$D=\min\left\{\left. \frac{d^2t_{\alpha\beta\gamma}}{s_\gamma} \ \right|\ \alpha,\beta,\gamma\in[m], s_\gamma\neq 0\right\}.$$ Note that $D>0$, as $H$ is fully connected. 
		For $\ell=1$, we have $P(a)>0=D(1-1)r$. Now let $\ell\in \mathbb{N}$ be arbitrary. We apply \eqref{lem:iterprob:3} and  \eqref{lem:iterprob:4} several times in order to establish $P(a^{\ell+1}) \geq D\ell r^{\ell+1}$: 
		\begin{align*}
			P(a^{\ell+1}) &= \sum_{\alpha \in [m]} s_\alpha p_{\alpha,\ell +1}\\
			&= \sum_{\alpha \in [m]} s_\alpha \sum_{\beta,\gamma\in [m]} t_{\alpha\beta\gamma} \sum_{i=1}^{\ell}p_{\beta,i}p_{\gamma,\ell+1-i}\\
			& = \sum_{i=1}^{\ell}\sum_{\alpha \in [m]}\sum_{\beta,\gamma\in [m]} (s_\alpha   p_{\beta,i})(t_{\alpha\beta\gamma}p_{\gamma,\ell+1-i})\\
			& \geq \sum_{i=1}^{\ell}\sum_{\alpha \in [m]}\sum_{\beta,\gamma\in [m]} (s_\alpha   p_{\beta,i})(Dd^{-2}s_\gamma p_{\gamma,\ell+1-i})\\
			& \geq \frac D{d^2}\sum_{i=1}^{\ell}\sum_{\alpha \in [m]}\sum_{\beta\in [m]} (s_\alpha   p_{\beta,i})\sum_{\gamma\in [m]} (s_\gamma p_{\gamma,\ell+1-i})\\
			& \geq \frac D{d^2}\sum_{i=1}^{\ell}\sum_{\alpha \in [m]} (s_\alpha   p_{\alpha,i})\sum_{\gamma\in [m]} (s_\gamma p_{\gamma,\ell+1-i})\\
			& = \frac D{d^2}\sum_{i=1}^{\ell}P(a^i)P(a^{\ell+1-i})\\
			& \geq \frac D{d^2}\sum_{i=1}^{\ell}dr^i dr^{\ell+1-i}\\
			& =  D\sum_{i=1}^{\ell}r^{\ell+1}\\
			&=D\ell r^{\ell+1},
		\end{align*}
		as required.		
	\end{proof}

	We conclude with a final remark on likelihood calculations.
	
	\begin{remark}\label{rmk:likelihoods}
		The practical application of probabilistic languages is as follows. Let $\underline{w}=(w_1,\ldots,w_j)\in (\Sigma^+)^j$ be a string of finite sequences over the alphabet $\Sigma$. Given a probabilistic language $P$, the likelihood that this probabilistic language generated this string is simply the product \begin{align}P(w_1)\cdot \ldots \cdot P(w_j).\label{eq:likelihood}\end{align} Likewise, given a probabilistic grammar $G$, the likelihood that this grammar generated the sequence $\underline{w}$ is given by \eqref{eq:likelihood}, where $P$ is the probabilistic language generated by $G$.
		
		In application, as task would be to find a probabilistic grammar with the highest likelihood for generating the sequence $\underline{w}$. However, due to the hierarchy we have established in this work, a sequence generated by an $n$-gram model may also have been generated by an $(n+1)$-gram model which mimics the behaviour of and $n$-gram model (see \Autoref{ex:ngraminclusion}). In fact, for a finite string $\underline{w}$ generated by an $n$-gram model $G$, it is plausible that there are $(n+1)$-gram models $H$ having a higher likelihood for generating the string $\underline{w}$ than $G$.
		
		Due to our main result \Autoref{thm:main}, this effect can be limited if $n$-gram models are compared to fully connected PCFGs: For any $n$-gram model $G$ and any fully connected PCFG $H$, we have $P_G\neq P_H$. If $\underline{w}$ is a sufficiently long string generated by $G$, it approximates the probabilistic language generated by $G$. Thus, the original model $G$ will have a higher likelihood to have generated $\underline{w}$ than any fully connected PCFG. The same result holds if $\underline{w}$ was initially generated by a fully connected PCFG.
	\end{remark}

	\pagebreak
	
	\noindent \textbf{CRediT Authorship Contribution Statement:} 
	Lothar Sebastian Krapp: conceptualization (equal); 
    investigation (lead);  
    writing – original draft (lead); writing – review \& editing (equal). Remo Nitschke: conceptualization (equal); investigation (support); writing – original draft (support); writing – review \& editing (equal).
	\\\\
	\textbf{Funding:} This research was funded by the NCCR Evolving Language, Swiss National Science Foundation Agreement \# 51NF40\_180888.
	\\\\
	\textbf{Acknowledgments:} We thank Balthasar Bickel and Jannik Kochert for insightful discussions about probabilistic languages.
	\\\\
	\textbf{Conflict of Interest:} The authors declare no conflict of interest. The funders had no role in the design and conduct of the study; preparation, review, or approval of the manuscript; and decision to submit the manuscript for publication.
	\\\\
	\textbf{Data Availability:} Not applicable.
	\\\\
	\textbf{Code Availability:} Not applicable.
	\\\\
	\textbf{Ethical Approval:} Not applicable.
	\\\\
	\textbf{Consent to Participate:} Not applicable.
	\\\\
	\textbf{Consent for Publication:} Not applicable. This work does not include data or images that require permission to be published.

\begin{figure}[b] 
\centering
\begin{minipage}[b]{.45\textwidth}
\centering
		\begin{tabular}{lcr||lcr}
        G$_{\textrm{PCFG}}$ &  &  & G$_{\textrm{PRG}}$ &  &  \\
        \hline
         S$\rightarrow$ & a & $0.5$ & S$\rightarrow$ & a & $0.5$  \\
         S$\rightarrow$ & S,S & $0.5$ & S$\rightarrow$ & a,S & $0.5$\\
         \hline
        \end{tabular}
        
        \bigskip
\noindent\textbf{A}
\end{minipage}\hfill
\begin{minipage}[b]{.45\textwidth}
\centering
	\includegraphics[width=0.9\linewidth]{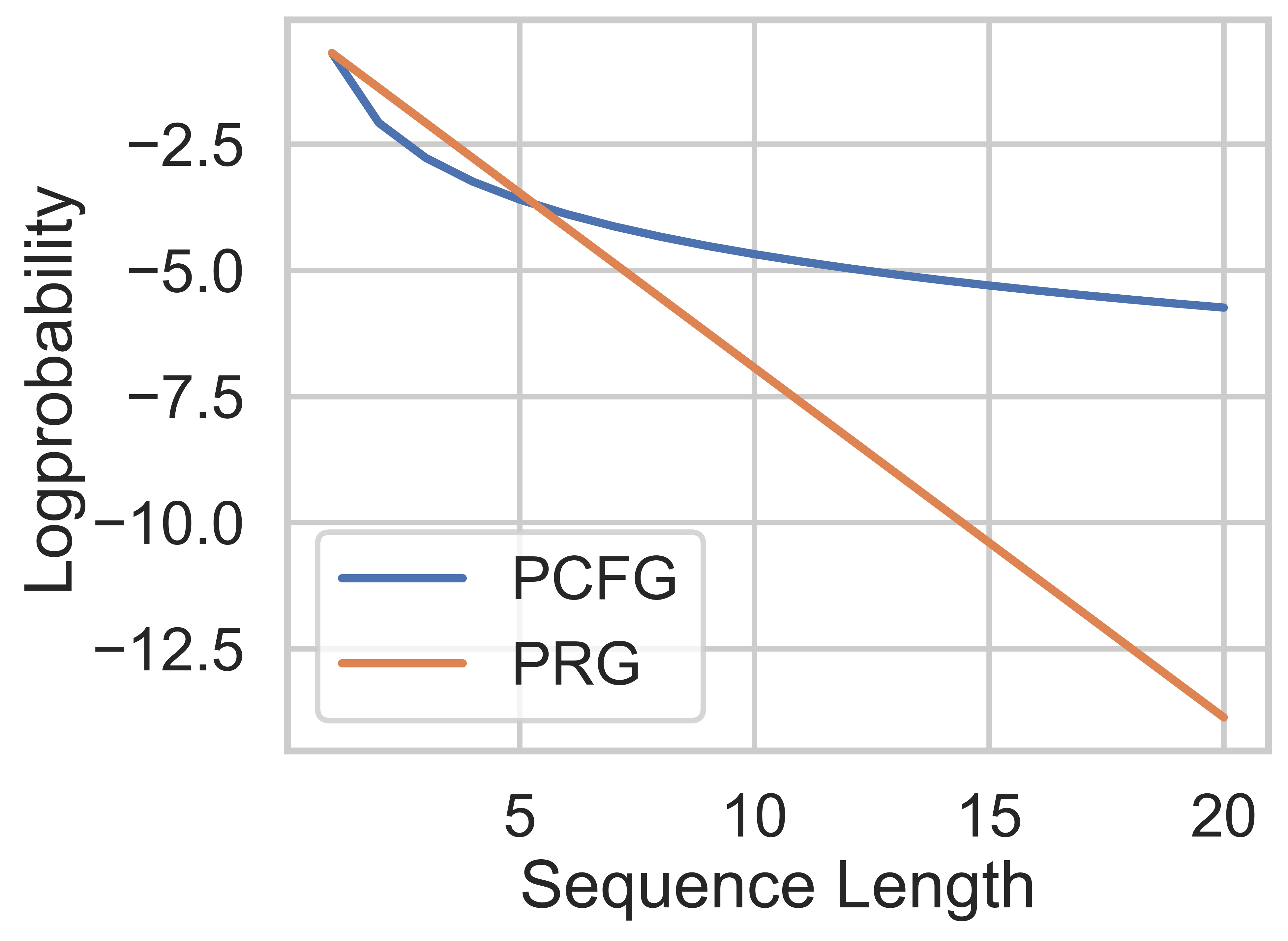} 
    
    \noindent\textbf{B}
	
\end{minipage}

\begin{equation*}
\begin{split}
  &p_{\textrm{pcfg}}(\{a_n\})=0.5^{2n-1}\cdot \frac{(2(n-1))!}{n!(n-1)!}\\
  &p_{\textrm{prg}}(\{a_n\})=0.5^n
\end{split}
\end{equation*}
        \bigskip
\noindent\textbf{C}
	
	\caption{\textbf{A distinct probability function is enforced by a fully connected PCFG.}
		This figure demonstrates how a fully connected PCFG models a distinct probability distribution from a probabilistic regular grammar. (\textbf{A}) defines a simple fully connected PCFG (G$_{\textrm{PCFG}}$) and a simple probabilistic regular grammar (G$_{\textrm{PRG}}$) with the associated probabilities of each rule.  (\textbf{B}) plots how the log-probability of a given sequence length $n$ decays for each grammar. (\textbf{C}) defines the probability functions of either grammar.}\label{eq:prob_func}
	\label{fig:pcfg_preg_dist} 
\end{figure}
    
\end{document}